\documentclass[12pt,a4paper]{article}
\usepackage[centertags]{amsmath}
\usepackage{amsfonts,amsthm,amssymb}
\usepackage{amssymb}
\usepackage{amsmath}
\usepackage{graphicx}
\usepackage{enumerate}
\usepackage[square]{natbib}

\theoremstyle{definition}

\newtheorem{corollary}{Corollary}

\newtheorem{proposition}{Proposition}

\newtheorem{axiom}{Axiom}
\newtheorem*{property}{Property}

\begin{document}

\title{Axiomatic Approach to Solutions of Games}
\author{Yakov Babichenko\footnote{Center for the Mathematics of Information, Department of Computing and Mathematical Sciences, California Institute of Technology. e-mail:babich@caltech.edu.}}

\maketitle

\begin{abstract}
We consider solutions of normal form games that are invariant under strategic equivalence. We consider additional properties that can be expected (or be desired) from a solution of a game, and we observe the following:
\begin{itemize}
\item Even the weakest notion of \emph{individual rationality} restricts the set of solutions to be equilibria. This observation holds for \emph{all} types of solutions: in pure-strategies, in mixed strategies, and in correlated strategies where the corresponding notions of equilibria are pure-Nash, Nash and coarse-correlated.
\end{itemize}
An action profile is \emph{(strict) simultaneous maximizer} if it simultaneously globally (strictly) maximizes the payoffs of all players. 

\begin{itemize}
\item If we require that a simultaneous maximizer (if it exists) will be a solution, then the solution contains the set of pure Nash equilibria.

\item There is no solution for which a strict simultaneous maximizer (if it exists) is the unique solution.
\end{itemize}
\end{abstract}

\section{Introduction}

A \emph{solution} of a strategic interaction suggests a subset of behaviors as possible behaviors in the interaction. The central solutions of strategic interactions are \emph{Nash equilibria and correlated equilibria} which are defined to be the set of all behaviors where no player can gain by a unilateral deviation. Equilibrium behaviors have well known and well studied undesirable properties: non-uniqueness, inefficiency, the underlying assumption that each player knows opponents' behavior, and others... This motivates the search for a "better" solution which overcomes part of the undesirable properties. This issue was addressed in areas as equilibrium selection (see \citep{hs}) and learning dynamics (see \citep{hm}). Alternative solutions was suggested: risk-dominant equilibrium (see \citep{hs}), scholastically stable equilibrium (see \citep{y}),  minimal curb (see \citep{bw}), sink equilibrium (see \citep{gmv}) and many others.


This note suggests reasonable basic properties of a solution (axioms). We show that these basic properties characterize the classical notions of Nash equilibrium and correlated equilibrium; In other words, these equilibria are the only possible solutions that satisfies these basic properties. 

The axioms that are considered in this note are the following:

\begin{itemize}
\item Strategic equivalence (SE), which requires that the set of solutions of two strategically equivalent games will be identical.

\item Individual rationality (IR). We consider the weakest notion of individual rationality that requires that in a solution, every player receives a payoff that is above the payoff that he can guarantee by pure actions.

\item Simultaneous maximizer (SM), which requires that if there exists an action profile where all players receive their best possible payoff, then this action profile will be a possible solution.

\item Unique Simultaneous maximizer (USM), which requires that if there exists an action profile where all players receive strictly their best possible payoff, then this action profile will be a unique solution.
\end{itemize}

Our main observations are the following:

\begin{enumerate}
\item Every solution that satisfies (SE) and (IR) is a subset of pure Nash/ Nash/ coarse-correlated equilibria, depending on the type of allowed behaviors, pure/ mixed/ correlated strategies correspondingly (Corollary \ref{cor:sub} and Proposition \ref{theo:e-c}). This observation demonstrates how unreasonable could a non-equilibrium solution be (from the perspective of at least one player) if the solution is invariant under strategic equivalence.

\item Every solution that satisfies (SE) and (SM) contains the set of pure Nash equilibia (Propositions \ref{theo:pne} and \ref{theo:pne-c}). A consequence of this observation is that any equilibrium selection that satisfies strategic equivalence\footnote{Indeed most of studied models of of equilibrium selections satisfy strategic equivalence. For instance risk-dominance or stochastic stability.} will be inefficient in cases where efficiency seems reasonable (i.e., cases of existence of simultaneous maximazer).

\item By combining observations (1) and (2) we obtain that pure Nash equilibrium is the unique solution that satisfies axioms (SE), (IR), and (SM) (see Corollaries \ref{cor:ch-pne} and \ref{cor:ch-pne-c}), which provides a simple axiomatization for pure Nash equilibrium. 

\item There is no solution that satisfies (SE) and (USM) (Corollaries \ref{cor:ch-pne} and \ref{cor:ch-pne-c}). This observation demonstrates the difficulty in selecting the efficient equilibrium, even in cases where this selection seems reasonable (i.e., cases of existence of strict simultaneous maximazer). This difficulty was demonstrated  also in \citep{as}, where they essentially shows that such selection can be done in two player games, but the selecting mechanism is involved.

\end{enumerate}


Axiomatic approach to Nash equilibrium was addressed previously in Norde et al. \citep{nprv}, where they also obtain an axiomatization for pure Nash equilibrium (as we do in Corollaries \ref{cor:ch-pne} and \ref{cor:ch-pne-c}). The central axiom in \citep{nprv} is \emph{consistency} that was presented in \citep{pt} and \citep{ppt}. Consistency requires that if a subset of players $P$ act according to a solution $s$, then in the game between the players that are not in $P$ acting according to $s$ remains a solution. The nature of this axiom is to maintain the \emph{payoffs} in the game and to require invariance of the solution under a proper change of the \emph{structure of the game}. Note also that the consistency axiom, by its nature, cannot be defined for correlated behaviors. In the present note the central axiom is strategic equivalence which has the opposite nature: we maintain the \emph{structure of the game} and require invariance of the solution under a proper change of the \emph{payoffs}. Our axioms are well defined for correlated strategies and indeed Propositions \ref{theo:e} and \ref{pro:cor} proves axiomatic characterization for correlated and coarse-correlated equilibria. Axiomatic approach to equilibrium with consistency as a central axiom was discussed also in \citep{vvb} where they consider games with vector-payoffs, in \citep{hptb} where they consider Bayesian games, and in \citep{vkn} where they obtain axiomatization for minimal curbs.

\section{Preliminaries}

We use standard notation of a normal form game. The set of \emph{players} is $[n]$. The set of \emph{actions of player $i$} is denoted by $A_i$, and $A=\times_{i\in [n]} A_i$ denotes the \emph{action profiles} set. The set of probability distributions over a set $B$ is denoted by $\Delta(B)$. The \emph{payoff function} of player $i$ is denoted by $u_i:A\rightarrow \mathbb{R}$, and it is miltylinearily extended to $u_i:\Delta(A)\rightarrow \mathbb{R}$. We will identify the game with its payoff profile $u=(u_i)_{i\in [n]}$. The set of all games with action profiles set $A$ is denoted by $\Gamma(A)$.

For a distribution $x\in \Delta(A)$ we denote by $x_i \in \Delta (A_i)$ and $x_{-i}\in \Delta(A_{-i})$ the marginal distributions over the elements $A_i$ and $A_{-i}$ correspondingly. For an action $b_i\in A_i$ we denote by $x|b_i \in \Delta(A_{-i})$ the conditional distribution of $x$ over $A_{-i}$ given the event $\{a:a_i=b_i\}$. The notion $x|b_i$ is well defined only if the event $\{a:a_i=b_i\}$ has positive probability (i.e., if $x_i(b_i)>0$).

For two game $u,v$ with the same actions set the game $u+v$ is defined by $(u+v)(a)=u(a)+v(a)$.

We denote by $z(c,i,b_{-i})$ the following game,
\begin{equation}\label{eq:d}
(z(c,i,b_{-i}))_j(a)=\begin{cases}
c& \text{ if } j=i \text{ and } a_{-i}=b_{-i}\\
0& \text{ otherwise.}
\end{cases}
\end{equation}
where player $i$ receives a payoff of $c$ iff his opponents play $b_{-i}$. 

The set of possible \emph{behaviors} in a game with an actions set $A$ is denoted by $\mathcal{B}=\mathcal{B}(A)$ and it may vary upon the model we consider. We can consider pure strategies behaviors ($\mathcal{B}=A$), mixed strategies behaviors ($\mathcal{B}=\times_{i\in [n]} \Delta A_i$), or correlated strategies behaviors ($\mathcal{B}=\Delta(A)$). 

A \emph{solution} of a game is a mapping $S:\Gamma(A)\twoheadrightarrow \mathcal{B}(A)$, which maps every game into a set of possible solutions. For every solution we should specify the type of behaviors it produces (pure, mixed, or correlated). We will not specify the type of behaviors in cases where the statement holds for all tree types of behavior.


The central solutions of normal form games are: pure Nash equilibrium, Nash equilibrium, correlated equilibrium and coarse-correlated equilibrium.

Pure Nash equilibrium is a pure-strategies solution that is defined by:
\begin{equation*}
PNE(u)=\{a\in A : u_i(a)\geq u_i(b_i,a_{-i}) \text{ for all } i\in [n] \text{ and all } b_i \in A_i \}.
\end{equation*}
Nash equilibrium is a mixed-strategies solution that is defined by:
\begin{equation*}
NE(u)=\{x\in \times_{j\in [n]} \Delta (A_j) : u_i(x)\geq u_i(b_i,x_{-i}) \text{ for all } i\in [n] \text{ and all } b_i \in A_i \}.
\end{equation*}
Correlated equilibrium is a correlated-strategies solution that is defined by:
\begin{eqnarray*}
CE(u)=\{x\in \Delta (A) :& u_i(a_i,x|a_i)\geq u_i(b_i,x|a_i) \text{ for all } i\in [n], \\
&\text{ all } a_i \text{ s.t. } x_i(a_i)>0, \text{ and all } b_i \in A_i \}.
\end{eqnarray*}
Coarse-correlated equilibrium is a correlated-strategies solution that is defined by:
\begin{equation*}
CCE(u)=\{x\in \Delta (A) : u_i(x)\geq u_i(b_i,x|a_i) \text{ for all } i\in [n] \text{ and all } b_i \in A_i \}.
\end{equation*}
A farther discussion on the behavior-model difference between correlated equilibrium and coarse-correlated equilibrium appears in Section \ref{sec:cor}. 

Our central axiom for a solution is invariance under  strategic equivalence.

\begin{axiom}
\textbf{Strategic Equivalence (SE)}: $S(u)=S(u+z(c,i,b_{-i}))$ for every $u,c,i$ and $b_{-i}$.
\end{axiom}

The rational behind the axiom is as follows. In the game $u+z(c,i,b_{-i})$ whether or not player $i$ receives an additional payoff of $c$ \emph{does not} depend on the behavior of players $i$. Therefore, if $x\in S(u)$ is a reasonable behavior in the game $u$ (i.e., a solution), so it should be for the game $u+z(c,i,b_{-i})$.

\section{Finite actions sets}\label{sec:r}

\subsection{Restriction to Equilibrium Behaviors}

We define the \emph{pure individually rational} value of player $i$ to be
\begin{equation*}
pir_i (u):=\underset{a_i\in A_i}{\max} \left( \underset{a_{-i}\in A_{-i}}{\min}  u_i(a_i,a_{-i}) \right).
\end{equation*}
Note that player $i$ can \emph{guarantee} a payoff of at least $pir_i$ by playing a pure strategy \emph{irrespective} of the strategy of his opponents. Therefore it is reasonable to assume that every player will get a payoff of at least $pir_i$ in a solution of the game:

\begin{axiom}
\textbf{Pure Individual Rationality (PIR)}: $u_i(x)>pir_i (u)$ for every $x\in S(u)$ and every player $i$. 
\end{axiom}

The rational behind axiom \textbf{(PIR)} relies on a deviation argument: if player $i$ receives a payoff below $pir_i$ then he will prefer to deviate to a pure strategy that guarantees him at least $pir_i$. We emphasize that the above deviation argument is significantly weaker than the deviation argument in the definition of equilibrium, because in an equilibrium player $i$ should also know the strategy of his opponents. 

There are other notions of individual-rationality: individual rationality in mixed strategies where the opponents are not allowed to correlate their punishment strategies (i.e., $ir_i (u):=\max_{x_i\in \Delta (A_i)} \min_{x_{-i}\in \times_{j\neq i} \Delta (A_j)}  u_i(x_i,x_{-i})$), and individual-rationality where we allow to the opponents to correlate their punishment strategy (i.e., $cir_i (u):=\max_{x_i\in \Delta (A_i)} \min_{x_{-i}\in \Delta (A_{-i})}  u_i(x_i,x_{-i})$). Depending on the type of solution (in pure, mixed or correlated strategies) each one of the individually rational levels might be a reasonable axiom for the solution. However, the pure individual-rationality axiom is the \emph{weakest} axiom among the three, and we will see that even this weakest axiom is sufficient to eliminate all the non-equilibria behaviors in all three cases of behavior (pure, mixed, and correlated).

We would like to treat to all the cases of behavior in a single statement. In order to do so, we consider the property that defines the notion of equilibrium in all cases:

\begin{property}
\textbf{No Unilateral Deviation Improvement (NUDI)}: $u_i(x)\geq u_i(a_i,x_{-i})$ for every $x\in S(u)$, every $i\in [n]$, and every $a_i\in A_i$.
\end{property}

For the case of pure-strategy behaviors, \textbf{(NUDI)} defines the set of pure Nash equilibria. For the case of mixed-strategy behaviors, \textbf{(NUDI)} defines the set of mixed Nash equilibria. For the case of correlated-strategy behaviors \textbf{(NUDI)}, defines the set of coarse-correlated equilibria.

The following proposition states that the weak axiom \textbf{(PIR)} is sufficient for eliminating all the non-equilibrium behaviors.

\begin{proposition}\label{theo:e}
Every solution $S$ that satisfies axioms \textbf{(SE)} and  \textbf{(PIR)} satisfies also property \textbf{(NUDI)}.
\end{proposition}

Straightforward corollary from the proposition is the following.

\begin{corollary}\label{cor:sub}
If a pure-strategies/ mixed strategies/ correlated-strategies solution $S$ satisfies axioms \textbf{(SE)} and \textbf{(PIR)}, then $S(u)\subseteq PNE(u)$/ $S(u)\subseteq NE(u)$/ $S(u)\subseteq CCE(u)$ (correspondingly).
\end{corollary}

A consequence of the corollary is that any solution that for some games contains non-equilibrium behaviors, must violate either strategic equivalence or the pure individual rationality.




\begin{proof}[Proof of Proposition \ref{theo:e}] 
Let $S$ be a solution that satisfies \textbf{(SE)} and \textbf{(PIR)} and assume that a behavior $x$ does not satisfy the condition of \textbf{(NUDI)}, which means that there exists player $i$ and action $a_i$ such that $u_i(a_i,x_{-i})> u_i(x)$, we show that $x\notin S(u)$. We construct a strategically equivalent game $v$, where we set all payoffs of player $i$ when he plays $a_i$ to be 0. We define the game $v$ by
\begin{equation*}
v=u+\underset{b_{-i}\in A_{-i}}{\sum} z(-u_i(a_i,b_{-i}),i,b_{-i}).
\end{equation*}
It easy to see that indeed $v(a_i,b_{-i})=0$ for all $b_{-i}\in A_{-i}$, therefore $pir_i (v)\geq 0$. On the other hand note that 
\begin{equation*}
v(x)=u(x)-u(a_i,x_{-i})<0
\end{equation*}
because the payoff of player $i$ in the game $\sum_{b_{-i}} z(-u_i(a_i,b_{-i}),i,b_{-i})$ when his opponents are playing $x_{-i}$ is equal to $-u(a_i,x_{-i})$. Therefore by the \textbf{(PIR)} axiom we get $x\notin S(v)$ and by the \textbf{(SE)} axiom we get $S(v)= S(u)$.
\end{proof}

\subsubsection{Correlated equilibrium}\label{sec:cor}

Once we allow correlated strategies behavior we have to assume existence of a correlation device. One of the suggested models for a correlation device (and probably the most central one) assumes existence of a \emph{mediator}. The \emph{mediator} randomizes an action profile $a$ according to a correlated distribution $x\in \Delta(A)$ and then \emph{recommends} to each player $i$ his action $a_i$. There are two reasonable models of correlated behavior.

\begin{enumerate}
\item Each player $i$ does not observe the recommendation of the mediator at the moment when he decides about his act. In such a case, possible acts of player $i$ are "follow the recommendation" or "deviate to action $a_i$". This is the model that leads to the definition of coarse-correlated equilibrium.

\item Each player $i$ observes the recommendation of the mediator before he decides about his action. In such a case possible acts of the player are "play $a'_i$ when the recommendation was $a_i$". This is the model that leads to the definition of correlated equilibrium.
\end{enumerate}

In the case of the second model, axiom \textbf{(PIR)} is too weak, because it does not take into account the additional information that player $i$ has before making his decision. A reasonable analog of axiom \textbf{(PIR)} in the case of model (2) is the following.

\begin{axiom}
\textbf{Pure Individual Rationality for All Recommendations (PIRAR)}: $u_i(a_i,(x|a_i)_{-i})>pir_i (u)$ for every $i\in [n]$ and every $a_i$ such that $x_i(a_i)>0$. 
\end{axiom}

The axiom states that for every possible recommendation to player $i$, he receives an expected payoff of at least $pir_i (u)$ by following the recommendation of the mediator. Note that in model (2) player $i$ can decide to deviate to a pure action that guarantee him $pir_i$ \emph{after} he observed the recommendation $a_i$. Therefore axiom \textbf{(PIRAR)} is indeed reasonable in the model (2) of correlated behavior.

The analog of Corollary \ref{cor:sub} for the case of correlated equilibrium is the following.

\begin{proposition}\label{pro:cor}
If a correlated-strategies solution $S$ satisfies axioms \textbf{(SE)} and \textbf{(PIRAR)}, then $S(u)\subseteq CE(u)$.
\end{proposition}

The proof is similar to the proof of proposition \ref{theo:e}

\begin{proof}
Let $S$ be a solution that satisfies \textbf{(SE)} and \textbf{(PIRAR)} and assume that a distribution $x\in \Delta(A)$ is not a correlated equilibrium, which means that there exists player $i$ and actions $a_i,a'_i$ such that $u_i(a'_i,(x|a_i))> u_i(a_i,(x|a_i))$, we show that $x\notin S(u)$. We construct a strategically equivalent game $v$, where we set all payoffs of player $i$ when he plays $a'_i$ to be 0. We define the game $v$ by
\begin{equation*}
v=u+\underset{b_{-i}\in A_{-i}}{\sum} z(-u_i(a'_i,b_{-i}),i,b_{-i}).
\end{equation*}
In the constructed game $pir_i (v)\geq 0$, and on the other hand $v(a_i,(x|a_i))=u(a_i,(x|a_i))-u(a'_i,(x|a_i))<0$. Therefore by axioms \textbf{(PIRAR)} and \textbf{(SE)} we get $x\notin S(v)=S(u)$.
\end{proof}

\subsection{Containment of Equilibrium Behaviors}

Non of the axioms \textbf{(SE)} or \textbf{(PIR)} requires that part of the behaviors will belong to the set of solutions. In particular,  the unreasonable solution $S(u)\equiv \emptyset$ satisfies both axioms \textbf{(SE)} and \textbf{(PIR)}. There are games where it is not obvious which behaviors should belong to the set of solutions. But in the case where there exists an action $a$ that simultaneously maximizes the payoffs of all players it is reasonable to assume that this action will be a possible solution.

\begin{axiom}
\textbf{Simultaneous Maximizer (SM)} If there exists $a\in A$ such that $u_i(a)\geq u_i(b)$ for all $i\in [n]$ and all $b\in A$, then $a\in S(u)$.
\end{axiom}

Under the strategic equivalence assumption, this basic axiom is sufficient for including all the pure Nash equilibria as possible solutions.

\begin{proposition}\label{theo:pne}
If a solution $S$ satisfies axioms \textbf{(SE)} and \textbf{(SM)}, then $S(u)\supseteq PNE(u)$ for every $u$.  
\end{proposition}

A consequence from Proposition \ref{theo:pne} is that every pure-Nash-equilibrium selection (i.e., a refinement of pure Nash equilibria) violates either strategic equivalence or the simultaneous maximizer axiom.

By combining Proposition \ref{theo:e} and \ref{theo:pne} we obtain an exact axiomatic characterization of pure Nash equilibria.

\begin{corollary}\label{cor:ch-pne}
Pure Nash equilibrium is the unique pure-strategies solution that satisfies axioms \textbf{(SE)},\textbf{(SM)}, and \textbf{(PIR)}. 
\end{corollary}

\begin{proof}[Proof of Proposition \ref{theo:pne}]
Let $S$ be a solution that satisfies axioms \textbf{(SE)} and \textbf{(SM)}, let $a\in PNE(u)$, and let $L$ be a bound on the payoffs of $u$ (i.e., $|u_i(a)|<L$). We construct a strategically equivalent game $v$, where we increase by a lot the payoffs of all players at the action profile $a$. We define the game $v$ by
\begin{equation*}
v=u+\underset{i\in [n]}{\sum} z(2L,i,a_{-i}).
\end{equation*}
The action profile $a$ in the game $v$ is a simultaneous maximizer for all players, because for every $b\in A$ and every $i\in [n]$ we have
\begin{eqnarray*}
v_i(a)&=& u_i(a)+2L \geq u_i(b)+2L = v_i(b) \text{ for } b \text{ s.t. } b_{-i}=a_{-i}\\
v_i(a)&=& u_i(a)+2L \geq L \geq u_i(b) = v_i(b) \text{ for } b \text{ s.t. } b_{-i}\neq a_{-i}.
\end{eqnarray*}
By axiom \textbf{(SM)} we have $a\in S(v)$ and by axiom \textbf{(SE)} we $S(v)=S(a)$.  
\end{proof}

A desirable property of a solution is uniqueness. Obviously it is unreasonable to assume that the solution will be unique for all games (e.g., in the game where all playoffs are 0, it is reasonable to allow all behaviors). But games that contain an action profile that simultaneously \emph{strictly} maximizes the payoffs of all players it is reasonable to desire that this maximizer will be the unique solution. 

\begin{axiom}
\textbf{Unique Simultaneous Maximizer (USM)}: If there exists $a\in A$ such that $u_i(a)>u_i(b)$ for all $i\in [n]$ and all $b\neq a$, then $S(u)=\{a\}$.
\end{axiom}

Under the strategic equivalence assumption, the desirable property $\textbf{(USM)}$ is unachievable. 

\begin{proposition}\label{pro:imp}
There is no solution $S$ that satisfies axioms \textbf{(SE)} and \textbf{(USM)}.
\end{proposition}

\begin{proof}[Proof of Proposition \ref{pro:imp}]
Assume by a contrary that such a solution exists. Take any game $u$ with payoffs bounded by $[-1,1]$ and with at least two strict pure Nash equilibria $a$ and $a'$. Consider the game $v=u+\sum_i z(3,i,a_{-i})$. By similar arguments to those in the proof of Proposition \ref{theo:pne} we obtain that $a$ is a unique strict maximizer. By axioms \textbf{(SE)} and \textbf{(USM)} we obtain $S(u)=S(v)=\{a\}$. By repeating the same arguments for the strict pure Nash equilibrium $a'$ we obtain $S(u)=\{a'\}$, which is a contradiction. 
\end{proof}

\section{Continuum Action Sets}

Corollary \ref{cor:ch-pne} provides an exact axiomatic characterization of pure Nash equilibria. However, for the case of finite actions sets, pure Nash equilibria is a problematic solution because it might be empty. On the other hand, for the case where the action set of each player is a compact convex set, the payoffs are continuous, and the payoff of each player is convex (with respect to his own action), pure Nash equilibrium is guaranteed to exist (see \citep{g}). In this section we show that the results of section \ref{sec:r} can be extended to the above mentioned settings.

We use the same notations as in section \ref{sec:r}. Now, $A_i$ is a compact convex set, the payoff function $u_i:A_i \times A_{-i} \rightarrow \mathbb{R}$ is continuous, and $u_i(\cdot, a_{-i})$ is a convex function of $a_i$ for every $a_{-i}$. In this section we will consider only pure-strategies behaviors (i.e., $S:\Gamma(A)\twoheadrightarrow A$).

Note that we cannot use the same axiom \textbf{(SE)} for invariance under strategic equivalence, because in the new game $u+z(c,i,b_{-i})$ the payoff of player $i$ is no longer continuous for $c\neq 0$. The analog of axiom \textbf{(SE)} for the case of continuous payoffs is defined as follows:

For a function $f:A_{-i} \rightarrow \mathbb{R}$ we denote by $z(i,f)$ the game where player's $i$ payoff is $(z(i,f))(a)=f(a_{-i})$ (irrespective of his own action), and the payoffs of all players $j\neq i$ is always $0$.

\begin{axiom}
\textbf{Continuous Strategic Equivalence (CSE)}: For every game $u$, every player $i\in [n]$, and every continuous function $f:A_{-i} \rightarrow \mathbb{R}$ holds $S(u)=S(u+z(i,f))$.
\end{axiom}

Note that the game $u+z(i,f)$ belong to the class of the discussed games: the payoffs are continuous (because $f$ is continuous, and the payoffs $(u+z(i,f))_i(\cdot,a_{-i})$ are convex because convexity is preserved under an addition of a constant.

The rational behind the axiom is exactly the same as for $\textbf{(SE)}$. The additional payoff of $f(a_{-i})$ that player $i$ receives does not depend on his own action.

The analog of Corollary \ref{cor:ch-pne} is the following:

\begin{proposition}\label{theo:e-c}
If a pure-strategy solution $S$ satisfies axioms \textbf{(CSE)} and \textbf{(PIR)}, then $S(u)\subseteq PNE(u)$.
\end{proposition}

The proof is similar to the proof of Proposition \ref{theo:e}.
\begin{proof}
For an action profile $\overline{a}\notin PNE(u)$ there exists a player $i$ and an action $b_i$ such that $u_i(b_i,\overline{a}_{-i})>u_i(a)$. We define the function $f:A_{-i}\rightarrow \mathbb{R}$ by $f(a_{-i})=-u(b_i,a_{-i})$ (note that $f$ is continuous because $u$ is continuous), and we consider the game $v=u+z(i,f)$. In the game $v$ we have $v_i(b_i,a_{-i})=0$ for all $a_{-i}\in A_{-i}$, and therefore $pir_i(v)\geq 0$. On the other hand, $v(\overline{a})=u(\overline{a})-u(b_i,\overline{a}_{-i})$. By the \textbf{(PIR)} axiom $\overline{a}\notin S(v)$, and by the \textbf{(CSE)} axiom $S(v)=S(u)$.
\end{proof}

The analog of Proposition \ref{theo:pne} is the following:

\begin{proposition}\label{theo:pne-c}
If a pure-strategy solution $S$ satisfies axioms \textbf{(CSE)} and \textbf{(SM)}, then $S(u)\supseteq PNE(u)$.
\end{proposition}

Here we cannot follow the ideas of the proof of Proposition \ref{theo:pne} directly, because the \textbf{(CSE)} axiom requires that the strategic equivalence transformation will be continuous.
\begin{proof}
Let $\overline{a}\in PNE(u)$. For every player $i\in [n]$ we define a function $f_i:A_{-i} \rightarrow \mathbb{R}$ by $f_i(a_{-i})=-\max_{a_i \in A_i} u(a_i,a_{-i})$, note that $f_i$ is continuous because the maximum of continuous functions is continuous. Now we consider the game $v=u + \sum_{i\in [n]} z(i,f_i)$. In the game $v$ all the payoffs are weakly negative because $v_i(b)=u_i(b)-\max_{a_i \in A_i} u(a_i,b_{-i})\leq 0$. Moreover, the payoffs at the action profile $\overline{a}$ are equal to $0$ because $v_i(\overline{a})=u_i(\overline{a})-\max_{a_i \in A_i} u(a_i,\overline{a}_{-i})= 0$. Therefore $\overline{a}$ is simultaneous maximizer, so $\overline{a}\in S(v)=S(u)$.
\end{proof}

By combining Propositions \ref{theo:e-c} and \ref{theo:pne-c} we obtain an exact characterization of pure Nash equilibria in continuous convex games (where pure Nash equilibrium is guaranteed to exist):

\begin{corollary}\label{cor:ch-pne-c}
Pure Nash equilibria is the unique solution that satisfies axioms \textbf{(CSE)},\textbf{(PIR)}, and \textbf{(SM)}.
\end{corollary}

The impossibility result of Proposition \ref{pro:imp} also can be extended to the case of continuous games:

\begin{proposition}\label{pro:imp-c}
There is no solution that satisfies axioms \textbf{(CSE)} and \textbf{(USM)}.
\end{proposition}

\begin{proof}
Similar to the proof of Proposition \ref{pro:imp} we start with a continuous game $u$ that contain at least two strict pure Nash equilibria $\overline{a}$ and $a'$, and we construct a strategically equivalent game where $\overline{a}$ is a strict simultaneous maximizer. In order to do so we use the constructed game $v$ in the proof of Proposition \ref{theo:pne-c}. The action profile $\overline{a}$ is a \emph{weak} simultaneous maximizer in the game $v$, but we easily can define a new strategically equivalent game $v'$ where $\overline{a}$ will be a \emph{strict} simultaneous maximizer: we define $g_i: A_{-i} \rightarrow \mathbb{R}$ to be any continuous function with unique global maximum at $\overline{a}_{-i}$ that is equal to $0$ (for instance, $g_i(a_{-i})=-||a_{-i}-\overline{a}_{-i}||_2$), and we define $v'=v+\sum_i z(i,g_i)$. Note that in the game $v'$ we have $v'_i(\overline{a})=0$ for all the players. In addition, for $a$ such that $a_{-i}=\overline{a}_{-i}$ we have $v'_i(a)<v'_i(\overline{a})$ because $\overline{a}$ is a strict pure Nash equilibrium, and for $a$ such that $a_{-i}\neq \overline{a}_{-i}$ we have $v'_i(a)<v_i(a)\leq 0$. Therefore for every $a\neq \overline{a}$ and every player $i$ we have $v'_i(a)<0$. Hence, $\overline{a}$ is indeed a strict simultaneous maxinizer. By  axioms \textbf{(USM)} and \textbf{(CSE)} we have $\{\overline{a}\}=S(v')=S(v)=S(u)$, but we can repeat the same arguments for the strict pure Nash equilibrium $a'$ and obtain $\{a'\}=S(u)$, which is a contradiction.

\end{proof}

\bibliographystyle{plainnat}
\bibliography{bib}

\end{document}